\documentclass[12pt]{article}
\usepackage{amssymb,amsthm,amsmath,latexsym}
\usepackage{url}
\usepackage{lscape}
\newtheorem{thm}{Theorem}[section]
\newtheorem{prop}[thm]{Proposition}
\newtheorem{lem}[thm]{Lemma}
\newtheorem{cor}[thm]{Corollary}

\theoremstyle{remark}
\newtheorem{rem}[thm]{Remark}
\newcommand{\FF}{\mathbb{F}}

\newcommand{\0}{\mathbf{0}}
\newcommand{\1}{\mathbf{1}}
\newcommand{\ww}{\omega}
\newcommand{\vv}{\overline{\omega}}

\newcommand{\cC}{\mathcal{C}}

\DeclareMathOperator{\wt}{wt}

\begin{document}
\title{Construction of binary LCD codes, ternary LCD codes and
quaternary Hermitian LCD codes}

\author{
Masaaki Harada\thanks{
Research Center for Pure and Applied Mathematics,
Graduate School of Information Sciences,
Tohoku University, Sendai 980--8579, Japan.
email: {\tt mharada@tohoku.ac.jp}.}
}

%\date{}

\maketitle

\begin{abstract}
We give two methods for constructing many
linear complementary dual (LCD for short) codes from a given
LCD code, by modifying some known methods for constructing self-dual codes.
Using the methods, we construct binary LCD codes and quaternary
Hermitian LCD codes, which improve the
previously known lower bound on the largest minimum weights.
\end{abstract}

%%%%%%%%%%%%%%%%%%%%%%%%%%%%%%%%%%
\section{Introduction}\label{Sec:1}

Let $\FF_q$ denote the finite field of order $q$,
where $q$ is a prime power.
An $[n,k]$ code $C$ over $\FF_q$ is called 
{\em (Euclidean) linear complementary dual}
if $C \cap C^\perp = \{\0_n\}$, where 
$C^\perp$ is the dual code of $C$ and 
$\0_n$ denotes the zero vector of length $n$.
An $[n,k]$ code $C$ over $\FF_{q^2}$ is called 
{\em Hermitian linear complementary dual}
if $C \cap C^{\perp_H} = \{\0_n\}$,
where $C^{\perp_H}$ is the Hermitian dual code of $C$.
These two families of codes are collectively called
{\em linear complementary dual} (LCD for short) codes.
%A code $C$ over $\FF_{q^2}$ is called 
%{\em Hermitian self-dual}
%if $C = C^{\perp_H}$

LCD codes were introduced by Massey~\cite{Massey} and gave an optimum linear
coding solution for the two user binary adder channel.
% Recently, much work has been done concerning LCD codes
% for both theoretical and practical reasons.
In recent years, there has been an increasing interest in
LCD codes for both theoretical and practical reasons.
In particular, Carlet, Mesnager, Tang, Qi and Pellikaan~\cite{CMTQP}
showed that 
any code over $\FF_q$ is equivalent to some LCD code
for $q \ge 4$ and
any code over $\FF_{q^2}$ is equivalent to some Hermitian LCD code
for $q \ge 3$.
This motivates us to study binary LCD codes, ternary LCD codes
and quaternary Hermitian LCD codes.

Meanwhile, self-dual codes are one of the most interesting classes of
codes.
For example, 
this interest is justified by many combinatorial objects 
and algebraic objects related to self-dual codes
(see e.g.\ \cite{RS-Handbook} for basic facts concerning 
self-dual codes).
Many methods for constructing self-dual codes are known.
Although
% the definition of self-dual codes is far from that of LCD codes,
% the definitions of self-dual codes and LCD codes
% are quite different,
% are far apart,
the definitions say that self-dual codes and LCD codes
are quite different classes of codes, 
codes of both classes are characterized by their generator matrices.
% In this direction,
In this situation,
it is natural to consider methods for constructing
LCD codes by modifying some known methods of self-dual codes.
% In particular, we are interested in methods of constructing
% many self-dual codes from a given self-dual code.
%%
Also,
it is a fundamental problem to determine the largest minimum
weight among all $[n,k]$ codes in a certain class of codes for a given pair $(n,k)$.
In this paper, we give two methods for constructing
many LCD codes from a given LCD code
by modifying known methods of constructing
many self-dual codes from a given self-dual code.
Using the methods, we construct binary LCD codes and quaternary Hermitian LCD
codes, which improve the
previously known lower bound on the largest minimum weights.

The paper is organized as follows.
In Section~\ref{sec:pre}, 
we give some definitions, notations and basic results used in this paper.
In Section~\ref{sec:method}, 
we give 
a method for constructing
binary LCD $[n+2,k+1]$ codes,
ternary LCD $[n+3,k+1]$ codes and 
quaternary Hermitian LCD $[n+2,k+1]$ codes
from a given
binary LCD $[n,k]$ code,
ternary LCD $[n,k]$ code and 
quaternary Hermitian LCD $[n,k]$ code, respectively.
This is established by modifying the method of self-dual codes in~\cite{H70}.
In addition, we give a method for constructing binary even LCD $[n,k]$ codes
from a given binary even LCD $[n,k]$ code.
This is established by modifying the methods of self-dual codes
in~\cite{H96} and~\cite{HK}.
In Section~\ref{sec:2}, 
% by the two methods in Section~\ref{sec:method},
we construct binary LCD codes with parameters:
\begin{align*}
&
[28,6,12],
[28,13,8],
[29,13,8],
[30,7,12],
[30,14,8],
[32, 19,6],
\\&
[32, 20,6],
[33,22,6],
[ 34, 10, 12 ],
[34,21,6],
[34, 22,6],
[36,6,16], 
\\&
[36,23,6],
[40,6,18] \text{ and } 
[40,8,16].
\end{align*}
These codes determine the largest minimum weights as follows:
\[
d_2(n,k)=
\begin{cases}
6 & \text{ if } (n,k)=(32,19), (32,20), (33,22), (34,21), \\
   & \qquad \qquad (34,22) \text{ and } (36,23), \\ 
8 & \text{ if } (n,k)=(28,13), (29,13) \text{ and }(30,14), \\
12 & \text{ if } (n,k)=(28,6), (30,7) \text{ and }(34,10),\\
16 & \text{ if } (n,k)=(36,6) \text{ and } (40,8),\\
18 & \text{ if } (n,k)=(40,6),
\end{cases}
\]
where $d_2(n,k)$ denotes the largest minimum weight among
all binary LCD $[n,k]$ codes.
%% \begin{align*}
%% d_2(32, 19)=
%% d_2(32, 20)=
%% d_2(33, 22)=
%% d_2(34,21)&\\=
%% d_2(34, 22)=
%% d_2(36,23)&=6,
%% \\
%% d_2(28,13) =d_2(29,13)=d_2(30,14)&=8,
%% \\
%% d_2(28,6) =d_2(30,7)=d_2(34,10)&=12,
%% \\
%% d_2(36,6)=d_2(40,8)&=16, %  \text{ and }
%% \\
%% d_2(40,6)&=18.
%% \end{align*}
In Section~\ref{sec:4}, 
we construct quaternary Hermitian LCD codes with parameters:
\begin{align*}
&
[21,8, 9],
[21,10,8], 
[21,11,7],
[22,8,10],
[ 23, 18, 4 ],
[ 24, 16, 6 ],
\\&
[25, 18,4],
[ 25, 19, 4 ],
[26, 16,6],
[ 26, 17, 6 ],
[26, 20,4],
[ 26, 21, 4 ],
\\&
[27, 21,4],
[ 27, 22, 4 ],
[28, 21,4],
[ 28, 22, 4 ],
[ 28, 23, 4 ],
[29, 22,4],
\\&
[ 29, 23, 4 ],
[30, 23,4] \text{ and }
[ 30, 24, 4 ].
\end{align*}
These codes improve the
previously known lower bounds on the largest minimum weights.
In particular, we have
\[
d^H_4(n,k)=
\begin{cases}
4 &\text{ if } (n,k)=
(23, 18),(26, 21),(27, 22),(28, 23) \text{ and } (30, 24),\\
6 &\text{ if } (n,k)=(24, 16),
\end{cases}
\]
where $d^H_{4}(n,k)$ denotes the largest minimum weight among
all quaternary Hermitian LCD $[n,k]$ codes.
%% \begin{align*}
%% &
%% d^H_4(23, 18)=
%% d^H_4(26, 21)=
%% d^H_4(27, 22)=
%% d^H_4(28, 23)=
%% d^H_4(30, 24)=4,
%% \\&
%% d^H_4(24, 16)=6 \text{ and }
%% d^H_4(22,8) =10.
%% \end{align*}
Also, the above  quaternary Hermitian LCD codes show the new existence of 
some entanglement-assisted quantum codes.
In Section~\ref{sec:3},
we give examples of ternary LCD codes constructed
by the first method in Section~\ref{sec:method},
which have minimum weights meeting the lower bounds on the largest
minimum weights among currently known all ternary codes.

%%%%%%%%%%%%%%%%%%%%%%%%%%%%%%%%
\section{Preliminaries}\label{sec:pre}

% Here, we give definitions, notations and 
% basic results used throughout this paper.
% Let $I_k$ denote the identity matrix of order $k$ and

Let $\FF_q$ denote the finite field of order $q$,
where $q$ is a prime power.
An $[n,k]$ code $C$ over $\FF_q$ is a $k$-dimensional vector 
subspace of $\FF_q^n$.
The parameter $n$ is called the {\em length} of $C$.
A matrix whose rows are linearly
independent and generate $C$ is 
called a {\em generator matrix} of $C$. 
The {\em weight} $\wt(x)$ of a vector $x$ of $\FF_q^n$
is the number of non-zero coordinates in $x$. 
The {\em minimum weight} of $C$ is defined as 
$\min\{\wt(x) \mid \0_n \ne x \in C\}$,
where $\0_n$ denotes the zero vector of length $n$.
An $[n,k,d]$ code over $\FF_q$ is an 
$[n,k]$ code over $\FF_q$ with minimum weight $d$. 
The elements of $C$ are called {\em codewords}.
Two codes $C$ and $C'$ over $\FF_q$ are {\em equivalent}
if there is a monomial matrix $P$  over $\FF_q$ with
$C' = C \cdot P$, where
$C \cdot P = \{ x P \mid x \in C\}$.

The {\em dual} code $C^{\perp}$ of an $[n,k]$ code $C$ 
over $\FF_q$ is defined as
$
C^{\perp}=
\{x \in \FF_q^n \mid \langle x,y\rangle = 0 \text{ for all } y \in C\},
$
where $\langle x,y\rangle = \sum_{i=1}^{n} x_i {y_i}$
for $x=(x_1,x_2,\ldots,x_n), y=(y_1,y_2,\ldots,y_n) \in \FF_q^n$.
For any element $\alpha$ of $\FF_{q^2}$, the {\em conjugation} of $\alpha$ is
defined as $\overline{\alpha}=\alpha^q$.
The {\em Hermitian dual} code $C^{\perp_H}$ of an $[n,k]$ code $C$ 
over $\FF_{q^2}$ is defined as
$
C^{\perp_H}=
\{x \in \FF_{q^2}^n \mid \langle x,y\rangle_H = 0 \text{ for all } y \in C\},
$
where $\langle x,y\rangle_H= \sum_{i=1}^{n} x_i \overline{y_i}$
for $x=(x_1,x_2,\ldots,x_n), y=(y_1,y_2,\ldots,y_n) \in \FF_{q^2}^n$.
A code $C$ over $\FF_q$ is called 
{\em (Euclidean) linear complementary dual}
if $C \cap C^\perp = \{\0_n\}$.
A code $C$ over $\FF_{q^2}$ is called 
{\em Hermitian linear complementary dual}
if $C \cap C^{\perp_H} = \{\0_n\}$.
These two families of codes are collectively called
{\em linear complementary dual} (LCD for short) codes.

Throughout this paper, we employ the following notations
of vectors and matrices.
Let $I_k$ denote the identity matrix of order $k$.
Let $\1_{k}$ denote the all-one vector of length $k$.
Let $A^T$ denote the transpose of a matrix $A$.
For a matrix $A=(a_{ij})$, 
the conjugate matrix of $A$ is defined as
$\overline{A}=(\overline{a_{ij}})$.

The following characterization is due to Massey~\cite{Massey}
(see e.g.\ \cite{GOS} for Hermitian LCD codes).

\begin{lem}\label{lem:LCD}
% Let $C$ be a code over $\FF_q$ (resp.\ $\FF_{q^2}$).  
% Let $G$  be a generator matrix of $C$.
% Then $C$ is LCD (resp.\ Hermitian LCD) if and only if
% $G G^T$ (resp.\ $G \overline{G}^T$) is nonsingular.
\begin{enumerate}
\item
Let $C$ be a code over $\FF_q$.  
Let $G$  be a generator matrix of $C$.
Then $C$ is LCD if and only if $G G^T$ is nonsingular.
\item
Let $C$ be a code over $\FF_{q^2}$.  
Let $G$  be a generator matrix of $C$.
Then $C$ is Hermitian LCD if and only if
$G \overline{G}^T$ is nonsingular.
\end{enumerate}
% Then the following properties are equivalent:
% \begin{enumerate}
% \item $C$ is LCD (resp.\ Hermitian LCD),
% \item $C^\perp$ is LCD  
% (resp.\ $C^{\perp_H}$ is Hermitian LCD),
% \item $G G^T$ (resp.\ $G \overline{G}^T$) is nonsingular.
% \end{enumerate}
\end{lem}

Codes over $\FF_2$, $\FF_3$ and $\FF_4$ are called {\em binary, ternary}
and {\em quaternary}, respectively.
A binary code $C$ is called {\em even} if the weights of all codewords of $C$
are even.
In this paper,
we denote the finite fields $\FF_2$ and $\FF_3$
by $\{0,1\}$ and $\{0,1,2\}$, respectively, and
we denote the finite field $\FF_4$ by
$\{0,1,\ww,\vv\}$, where $\vv= \omega +1$.
% In this paper,
% we take the elements of $\FF_2$ and $\FF_3$ to be
% $0,1$ and 
% $0,1,2$, respectively, and
% we take the elements of $\FF_4$ to be
% $0,1,\ww,\vv$, where $\vv= \omega +1$.

Throughout this paper, 
let $d_q(n,k)$ denote the largest minimum weight among
all LCD $[n,k]$ codes over $\FF_q$ $(q=2,3)$, and 
let $d^H_{4}(n,k)$ denote the largest minimum weight among
all quaternary Hermitian LCD $[n,k]$ codes.
The following bound was proved in~\cite[Theorem~8]{CMTQ} for
binary LCD codes and in~\cite[Theorem~1]{HS2}
for ternary LCD codes and quaternary Hermitian LCD codes.
We remark that the constructive proofs are given
in~\cite{CMTQ} and~\cite{HS2}.

\begin{lem}\label{lem:bound}
Suppose that $2 \le k \le n$.  Then
\[
d_q(n,k) \le d_q(n,k-1) \text{ for } q \in \{2,3\}
\text{ and }
d^H_4(n,k) \le d^H_4(n,k-1).
\]
\end{lem}

We end this section by giving
self-dual codes for the comparison with LCD codes.
A code $C$ over $\FF_q$ is called 
{\em (Euclidean) self-dual}
if $C = C^\perp$.
A code $C$ over $\FF_{q^2}$ is called 
{\em Hermitian self-dual}
if $C = C^{\perp_H}$.
These two families of codes are collectively called
{\em self-dual}  codes.
Self-dual codes are one of the most interesting classes of
codes (see e.g.\ \cite{RS-Handbook} for basic facts concerning 
self-dual codes).
It is known that a $[2n,n]$ code $C$ over $\FF_q$ (resp.\ $\FF_{q^2}$)
is self-dual (resp.\ Hermitian self-dual)
if and only if $GG^T=O_n$ (resp.\ $G{\overline G}^T=O_n$) 
for a generator matrix $G$ of $C$, where $O_n$ is the $n \times n$ zero matrix.
We emphasis that
both LCD codes and self-dual codes with generator matrices $G$ are characterized by 
$GG^T$ or $G{\overline G}^T$.

%%%%%%%%%%%%%%%%%%%%%%%%%%%%%%%%%%%%%%%%%%%%%%%%
\section{Methods for constructing LCD codes}\label{sec:method}

In this section, we give two methods for constructing many LCD codes
from a given LCD code.

%%%%%%%%%%%%%%%%%%%%%%%%%%%
\subsection{Binary LCD codes, ternary LCD codes
and quaternary Hermitian LCD codes}

Starting from a given self-dual $[2n,n]$ code,
some methods for constructing many self-dual $[2n+2,n+1]$ codes are
known 
(see e.g.\ \cite{H70} and~\cite{Kim}).
By modifying the method in~\cite{H70}, we give
the following method.

\begin{thm}\label{thm:I}
\begin{enumerate}
\item
Let $C_2$ be a binary LCD $[n,k]$ code with generator matrix $G_2$.
Let $x=(x_1,x_2,\ldots,x_n)$ be a vector of $\FF_2^{n}$.
Let $C_2(x)$ be the binary code with the
following generator matrix:
\[
G_2(x)=
%\left( \begin{array}{ccccccccccccc}
%1  & 0  & x_1 & \cdots & x_{n} \\
%y_1&y_1 & &     &   &         \\ 
%\vdots  & \vdots  &   &   G    &         \\ 
%y_k&y_k&   &         &        &         \\ 
%	\end{array}\right),
\left( \begin{array}{ccccccccccccc}
1  & 0  & x_1 & \cdots & x_{n} \\
\langle x, r_1\rangle&\langle x, r_1\rangle & &     &   &         \\ 
\vdots  & \vdots  &   &   G_2   &         \\ 
\langle x, r_k\rangle&\langle x, r_k\rangle&   &         &        &         \\ 
	\end{array}\right),
\]
where $r_i$ is the $i$-th row of $G_2$.
% and  $y_i=\langle x, r_i\rangle$ $(i=1,2,\ldots,k)$.
If $\wt(x)$ is even, then
$C_2(x)$ is a binary LCD $[n+2,k+1]$ code.

\item
Let $C_3$ be a ternary LCD $[n,k]$ code 
with generator matrix $G_3$.
Let $x=(x_1,x_2,\ldots,x_n)$ be a vector of $\FF_3^{n}$, and
let $a=(a_1,a_2,a_3)$ be a vector of $\FF_3^3$.
Let $C_3(x,a)$ be the ternary code with the
following generator matrix:
\[
G_3(x,a)=
%\left( \begin{array}{ccccccccccccc}
%a_1  & a_2 & a_3  & x_1 & \cdots & x_{n} \\
%y_1&y_1 &y_1& &     &   &         \\ 
%\vdots & \vdots & \vdots  &   &   G    &         \\ 
%y_k&y_k&y_k&   &         &        &         \\ 
%	\end{array}\right),
\left( \begin{array}{ccccccccccccc}
a_1  & a_2 & a_3  & x_1 & \cdots & x_{n} \\
2 \langle x, r_1\rangle&2 \langle x, r_1\rangle &2 \langle x, r_1\rangle& &     &   &         \\ 
\vdots & \vdots & \vdots  &   &   G_3    &         \\ 
2 \langle x, r_k\rangle&2 \langle x, r_k\rangle&2 \langle x, r_k\rangle&   &         &        &         \\ 
	\end{array}\right),
\]
where $r_i$ is the $i$-th row of $G_3$. 
%and $y_i=2 \langle x, r_i\rangle$ $(i=1,2,\ldots,k)$.
If either
$\wt(x) \not\equiv 2 \pmod 3$ and $a=(1,0,0)$
or 
$\wt(x) \not\equiv 0 \pmod 3$ and $a=(1,1,2)$,
then
$C_3(x,a)$ is a ternary LCD $[n+3,k+1]$ code.

\item
Let $C_4$ be  a quaternary Hermitian LCD $[n,k]$ code
with generator matrix $G_4$.
Let $x=(x_1,x_2,\ldots,x_n)$ be a vector of $\FF_4^n$.
Let $C_4(x)$ be the quaternary code with the
following generator matrix:
\[
G_4(x)=
%\left( \begin{array}{ccccccccccccc}
%1  & 0  & x_1 & \cdots & x_{n} \\
%y_1&y_1 & &     &   &         \\ 
%\vdots  & \vdots  &   &   G    &         \\ 
%y_k&y_k&   &         &        &         \\ 
%	\end{array}\right),
\left( \begin{array}{ccccccccccccc}
1  & 0  & x_1 & \cdots & x_{n} \\
\overline{\langle x, r_1\rangle_H}&\overline{\langle x, r_1\rangle_H} & &     &   &         \\ 
\vdots  & \vdots  &   &   G_4   &         \\ 
\overline{\langle x, r_k\rangle_H}&\overline{\langle x, r_k\rangle_H}&   &         &        &         \\ 
	\end{array}\right),
\]
where $r_i$ is the $i$-th row of $G_4$. 
%and $y_i=\overline{\langle x, r_i\rangle_H}$ $(i=1,2,\ldots,k)$.
If $\wt(x)$ is even, then
$C_4(x)$ is a quaternary Hermitian LCD $[n+2,k+1]$ code.
\end{enumerate}
\end{thm}
\begin{proof}
All cases are similar, and we give details only for (iii).

Let $r'_i$ denote the $i$-th row of $G_4(x)$.
From the construction of $G_4(x)$, it follows that
\begin{align*}
\langle r'_{1},r'_{1}\rangle_H & =1, \\
\langle r'_1,r'_i\rangle_H & =
\langle x, r_{i-1}\rangle_H+\langle x, r_{i-1}\rangle_H=0
\ (i=2,3,\ldots,k+1), \\
\langle r'_{i},r'_{i}\rangle_H  & =
\langle r_{i-1},r_{i-1}\rangle_H
\ (i=2,3,\ldots,k+1),\\
\langle r'_{i},r'_{j}\rangle_H & =
\langle r_{i-1},r_{j-1}\rangle_H
\ (i,j=2,3,\ldots,k+1, i \ne j).
\end{align*}
Thus, we have
 \[
G_4(x)\overline{G_4(x)}^T= 
\left( \begin{array}{ccccccccccccc}
1 & 0 &\cdots & 0 \\
0 & & \\
\vdots & &G_4\overline{G_4}^T &\\
0 & & \\
	\end{array}\right),
\]
which completes the proof of (iii).
\end{proof}

\begin{rem}
When $G_q$ $(q=2,3,4)$ has the form 
$\left( \begin{array}{cc}
I_k & A
\end{array}\right)$, 
we may assume without loss of generality that
$x_1=\cdots = x_k=0$.
This substantially reduces the number of the possible vectors $x$.
\end{rem}

In Sections~\ref{sec:2} and~\ref{sec:4},
we construct binary LCD codes and quaternary Hermitian LCD codes,
respectively, 
which improve the
previously known lower bounds on the largest minimum weights
by the above method.
Theorem~\ref{thm:I}
seems to be less useful for ternary LCD codes.
In Section~\ref{sec:3},
we give ternary LCD codes,
which have minimum weights meeting the lower bounds on the largest
minimum weights among currently known all ternary codes.

%%%%%%%%%%%%%%%%%%%%%%%%%%%
\subsection{Binary even LCD codes}
Now starting from a binary
self-dual $[2n,n]$ code with generator matrix of the form
$\left( \begin{array}{cc}
I_n & A
\end{array}\right)$, 
by transforming $A$, 
some methods for constructing many self-dual $[2n,n]$ codes are
known (see e.g.\ \cite{H96}, \cite{HK} and \cite{Tonchev}).
By modifying the methods in~\cite{H96} and \cite{HK}, we give
a method for constructing many binary even LCD codes from a given
binary even LCD code.

%%%%%%%%%%%%%%%%%%%%%%
\begin{thm}\label{thm:II}
Let $C$ be a binary even LCD  $[n,k]$ code with generator matrix 
$\left( \begin{array}{cc}
I_k & A
\end{array}\right)$.
Let $r_i$ be the $i$-th row of $A$.
Let $x$ be a vector of $\FF_2^{n-k}$.
Define a $k \times n-k$ 
matrix $A(x)$, where the $i$-th row $r'_i$ is defined as follows:
\[
r'_i = r_i + x + \langle r_i,x \rangle \1_{n-k}.
\]
Let $C(A,x)$ be the binary $[n,k]$ code with the following generator matrix: 
\[
\left( \begin{array}{cc}
I_k & A(x)
\end{array}\right).
\]
If both $\wt(x)$ and $n-k$ are even, 
then $C(A,x)$ is a binary even LCD $[n,k]$ code.
\end{thm}
\begin{proof}
Since $\wt(x)$ is even, we have $\langle x,x \rangle=0$ and $\langle x,\1_{n-k} \rangle=0$.
Since $n-k$ is even,  we have $\langle \1_{n-k},\1_{n-k} \rangle=0$.
Since $C$ is even, we have $\langle r_i,\1_{n-k} \rangle=1$.
Hence, we have
\begin{align*}
\langle r'_i,r'_j \rangle
&= \langle r_i + x + \langle r_i,x \rangle \1_{n-k}, r_j + x + \langle r_j,x \rangle \1_{n-k}\rangle 
\\ &= 
\langle r_i,r_j  \rangle
+\langle r_i, x \rangle
+\langle r_j,x \rangle
+\langle x,r_j \rangle
+\langle r_i,x\rangle
\\ &= 
\langle r_i,r_j  \rangle.
\end{align*}
This shows that
$
\left( \begin{array}{cc} I_k & A \end{array}\right)
\left( \begin{array}{cc} I_k & A \end{array}\right)^T
=
\left( \begin{array}{cc} I_k & A(x) \end{array}\right)
\left( \begin{array}{cc} I_k & A(x) \end{array}\right)^T
$.
The result follows.
\end{proof}

\begin{rem}
  If there is a binary even LCD $[n,k]$ code, then $k$
  must be even~\cite[Theorem~5]{CMTQ}.
\end{rem}

The above method constructs $2^{n-k-1}$
binary even LCD $[n,k]$ codes from a given
binary even LCD $[n,k]$ code.
It may be possible for the minimum weight of a binary even LCD $[n,k]$ code 
to be larger than that of a given binary even LCD $[n,k]$ code.
In this case, the minimum weight is increased by at least $2$.
In Section~\ref{sec:2}, we give such examples of binary even LCD codes,  
which improves the
previously known lower bounds on the largest minimum weights.
This illustrates the effectiveness of Theorem~\ref{thm:II}.

%%%%%%%%%%%%%%%%%%%%%%%%%%%%%%%%%%%%%%%%%%
\section{New binary LCD codes}
\label{sec:2}

A classification of binary LCD codes was done in \cite{AH-C} for 
$n \in \{1,2,\ldots,13\}$.
The largest minimum weights $d_2(n,k)$ were determined 
for $n \in \{1,2,\ldots,24\}$
(see~\cite[Table~1]{Galvez} for $n\in\{1,2,\ldots,12\}$,
\cite[Table~3]{HS} for $n\in\{13,14,15,16\}$
and~\cite[Table~15]{AH} for $n\in \{17,18,\ldots,24\}$).
For $n \in \{25,26,\ldots,30\}$, the current information on $d_2(n,k)$ can be
found in~\cite[Table~2]{FLFR}.
In this section, by Theorems~\ref{thm:I} and~\ref{thm:II},
we construct binary LCD codes,
which improve the
previously known lower bounds on the largest minimum weights $d_2(n,k)$.
All computer calculations in this section
were done with the help of {\sc Magma}~\cite{Magma}.

%%%%%%%
\subsection{Lengths $n$ with $25 \le n \le 30$}
By the {\sc Magma} function {\tt BestKnownLinearCode},
one can construct a binary $[27,12,8]$ code $C_{2,27}$.
The code $C_{2,27}$ has the following generator matrix:
\[
G_{2,27}=
\left( \begin{array}{ccccc}
            & 1 & 0 & \\
I_{12} & \vdots & \vdots & A'_{2,27}\\
            & 1 & 0 & \\
\end{array}\right),
\]
where $A'_{2,27}$ is listed in Figure~\ref{Fig:F2}.
Using Lemma~\ref{lem:LCD},
we verified by {\sc Magma} that $C_{2,27}$ is LCD\@. 
Let $C_{2,28}$ be the binary $[28,13]$ with the following generator matrix:
\[
G_{2,28}=
\left( \begin{array}{ccccccccccccc}
1&0&\cdots &0&0& 1 & 1& \cdots & 1 \\
0& &            &  &1    &0 &            &   &\\
\vdots&&I_{12}&& \vdots &\vdots &   &      A'_{2,27}      &   \\
0& &            &  &1    &0 &            &   &\\
  \end{array}\right).
\]
%Since each row of  $A'_{2,27}$ has even weight and
%$G_{2,27}G_{2,27}^T$ is nonsingular, 
%$G_{2,28}G_{2,28}^T$ is nonsingular.
Since each row of  $A'_{2,27}$ has even weight,
we have
\[
G_{2,28}G_{2,28}^T= 
\left( \begin{array}{ccccccccccccc}
1 & 0 &\cdots & 0 \\
0 & & \\
\vdots & & G_{2,27}G_{2,27}^T &\\
0 & & \\
	\end{array}\right).
\]
Thus, by Lemma~\ref{lem:LCD}, $C_{2,28}$ is LCD\@.  In addition, we 
verified by {\sc Magma}  that $C_{2,28}$ has minimum
weight $8$.
% The code $C_{2,28}$ has the following weight enumerator
% \begin{align*}
% &1+  275y^{8 }+  115y^{9 }+  586y^{10}+  350y^{11}+ 1135y^{12}
% +  893y^{13}+ 1260y^{14}
% \\&
% + 1380y^{15}+  628y^{16}+  893y^{17}+  170y^{18}+  350y^{19}
% +   41y^{20}+  115y^{21}.
% \end{align*}

\begin{lem}\label{lem:2-1}
There is a binary LCD $[28,13,8]$ code.
\end{lem}

%%%%%%%%%%%%%%%%%%%%%%%%%%%%%%%%%%%%%%%%%%%%%%%%%%%%
\begin{figure}[thbp]
\begin{center}
%{\small
{\footnotesize
%{\scriptsize
%{\tiny
\begin{align*}
A'_{2,27}=&
\left(
\begin{array}{ccccccccccccccccccccccccccccccc}
1&1&0&1&0&0&0&1&0&1&1&0&0\\
1&0&1&1&1&0&0&1&1&1&0&1&0\\
0&1&1&1&0&0&1&0&0&1&1&1&1\\
0&0&0&1&0&1&1&1&1&0&1&0&0\\
1&1&0&1&1&0&1&0&1&0&1&1&0\\
0&1&0&0&0&0&1&1&1&1&0&0&1\\
1&1&1&1&0&0&0&0&1&0&0&0&1\\
1&0&1&0&1&0&0&1&0&0&1&0&1\\
1&0&0&0&0&1&0&1&1&1&1&1&1\\
0&1&1&0&1&1&0&0&0&1&1&0&0\\
1&1&1&0&0&1&1&1&0&1&0&1&0\\
0&1&0&1&1&1&0&1&0&0&1&1&1
%1&0&1&1&0&1&0&0&0&1&0&1&1&0&0\\
%1&0&1&0&1&1&1&0&0&1&1&1&0&1&0\\
%1&0&0&1&1&1&0&0&1&0&0&1&1&1&1\\
%1&0&0&0&0&1&0&1&1&1&1&0&1&0&0\\
%1&0&1&1&0&1&1&0&1&0&1&0&1&1&0\\
%1&0&0&1&0&0&0&0&1&1&1&1&0&0&1\\
%1&0&1&1&1&1&0&0&0&0&1&0&0&0&1\\
%1&0&1&0&1&0&1&0&0&1&0&0&1&0&1\\
%1&0&1&0&0&0&0&1&0&1&1&1&1&1&1\\
%1&0&0&1&1&0&1&1&0&0&0&1&1&0&0\\
%1&0&1&1&1&0&0&1&1&1&0&1&0&1&0\\
%1&0&0&1&0&1&1&1&0&1&0&0&1&1&1
\end{array}
\right)
\\
A'_{2,28}=&
\left(
\begin{array}{ccccccccccccccccccccccccccccccc}
0&0&0&0&0&0&0&0&0&0&0&0&0&1&1&1&1&1&1&1&1&1\\
0&1&0&0&1&1&0&1&0&1&1&0&0&1&0&0&1&1&1&0&0&1\\
1&1&0&1&0&0&0&1&0&1&0&1&1&1&1&1&1&1&1&1&0&1\\
1&1&1&1&0&1&0&0&0&1&1&1&1&0&1&0&0&0&1&0&1&1\\
1&0&0&0&0&1&0&1&1&0&1&1&1&0&0&0&0&1&1&0&1&1\\
1&1&1&0&1&0&1&1&0&1&0&1&1&1&0&0&0&1&1&1&0&0\\
\end{array}
\right)
\end{align*}
}
\end{center}
\caption{Matrices $A'_{2,27}$ and $A'_{2,28}$}
\label{Fig:F2}
\end{figure}
%%%%%%%%%%%%%%%%%%%%%%%%%%%%%%%%%%%%%%%%%%%%%%%%%%%%

By applying Theorem~\ref{thm:I} to 
the generator matrices $G_{2,27}$ and $G_{2,28}$
of $C_{2,27}$ and $C_{2,28}$, 
it is possible to construct
$2^{15}$ binary LCD $[29,13]$ codes and
binary LCD $[30,14]$ codes, respectively.
In particular,
our computer search by {\sc Magma} found
a binary LCD $[29,13,8]$ code $C_{2,29}$ as $C_{2,27}(x)$ and 
a binary LCD $[30,14,8]$ code $C_{2,30}$ as $C_{2,28}(x')$,
where
\begin{align*}
x&=(0,\ldots,  0 ,1,1,1,0,0,1,1,1,1,0,0,1,0,0,0) \text{ and }\\
x'&=(0,\ldots,  0, 0,1,1,1,0,0,1,1,1,1,0,0,1,0,0,0).
\end{align*}
% The minimum weights of $C_{2,29}$ and $C_{2,30}$
% were calculated by {\sc Magma}.
% The codes $C_{2,29}$ and $C_{2,30}$ have the following weight enumerators
% \begin{align*}
% &  1+ 154y^{ 8}+ 115y^{ 9}+ 417y^{10}+ 350y^{11}+ 834y^{12}+ 893y^{13}
% +1215y^{14}+1380y^{15}
% \\&
% + 937y^{16}+ 893y^{17}+ 431y^{18}+ 350y^{19}+  90y^{20}
% + 115y^{21}+  17y^{22},
% \\&
% 1+  195y^{ 8}+  164y^{ 9}+  587y^{10}+  574y^{11}+ 1462y^{12}
% + 1544y^{13}
% + 2475y^{14}
% \\&
% + 2538y^{15}+ 2072y^{16}+ 1988y^{17}+ 1017y^{18}+  978y^{19}
% +  365y^{20}+  336y^{21}
% \\&
% +   17y^{22}+   70y^{23}+     y^{28},
% \end{align*}
% respectively.

\begin{lem}\label{lem:2-2}
There is a binary LCD $[29,13,8]$ code.
There is a binary LCD $[30,14,8]$ code.
\end{lem}

By using the method in~\cite[Section~6]{HS}, 
our computer search by {\sc Magma} found
a binary even LCD $[28,6,10]$ code $C'_{2,28}$
with generator matrix
$G'_{2,28}=\left( \begin{array}{cc}
I_6 & A'_{2,28}
\end{array}\right)$, where $A'_{2,28}$ is listed in Figure~\ref{Fig:F2}.
% The minimum weight was calculated by {\sc Magma}.
% Using Lemma~\ref{lem:LCD}, it was verified by {\sc Magma}
% that $C'_{2,28}$ is LCD\@.
By applying Theorem~\ref{thm:II} to the matrix $A'_{2,28}$, 
it is possible to construct
$2^{21}$ binary even LCD $[28,6]$ codes.
In particular,
our computer search by {\sc Magma} found
a binary even LCD $[28,6,12]$ code $D_{2,28}$ as $C'_{2,28}(A'_{2,28},x)$, where
\[
x=(1,0,1,1,1,1,1,0,\ldots,0).
\]
% The minimum weight of $D_{2,28}$ was calculated by {\sc Magma}.
%Note that the minimum weight of $D_{2,28}$ is larger than that of $C'_{2,28}$.
Note that the minimum weight of $D_{2,28}$ is increased by $2$.
% and this illustrates the effectiveness of Theorem~\ref{thm:II}.
%  an availability and validity of Theorem~\ref{thm:I}.
% 
% The code $D_{2,28}$ has the following weight enumerator
% \[
% 1+ 27y^{12}+ 18y^{14}+  6y^{16}+  9y^{18}+  2y^{20}+   y^{22}.
% \]

\begin{lem}\label{lem:2-4}
There is a binary even LCD $[28,6,12]$ code.
\end{lem}

Again, by applying Theorem~\ref{thm:I} to the generator matrix $G'_{2,28}$ of $D_{2,28}$, 
our computer search by {\sc Magma} found
a binary even LCD $[30,7,12]$ code $C'_{2,30}$ as $D_{2,28}(x)$, where
\[
x=(0,\ldots,0,1,0,1,0,0,0,1,0,1,1,1,1,0,1,1,1,0,0,1,0,1,0).
\]
% The minimum weight of $C'_{2,30}$ was calculated by {\sc Magma}.
% The code $D'_{2,28}$ has the following weight enumerator
% \begin{align*}
% &
% 1+14y^{12}+25y^{13}+ 22y^{14}+23y^{15}+13y^{16}
% + 9y^{17}+ 5y^{18}+ 5y^{19}+ 7y^{20}
% \\&
% + 2y^{21}+  y^{22}+  y^{24}.
% \end{align*}

\begin{lem}\label{lem:2-5}
There is a binary LCD $[30,7,12]$ code.
\end{lem}

%%%%%%%%% n >30 %%%%%%%%%%%%
%%%%%%%%% n >30 %%%%%%%%%%%%
%%%%%%%
\subsection{Lengths $n$ with $31 \le n \le 40$}

Now let us look at construction of binary  LCD $[n,k]$ codes for
$31 \le n \le 40$,
although no information on $d_2(n,k)$ is known for $n \ge 31$.
%Here we concentrate on the parameters $(n,k)=(36,6)$, $(40,6)$ and $(40,8)$.

\begin{lem}\label{lem:2-36}
For $(n,k,d) \in {\mathcal P}_{2}$, 
where 
\begin{equation}\label{eq:364040-2}
{\mathcal P}_{2}= \{(34,10,12),(36,6,16), (40,6,18), (40,8,16)\},
\end{equation}
there is a binary even LCD $[n,k,d]$ code.
%For $(n,k,d) \in {\mathcal P}_{2}$ in~\eqref{eq:364040-2},
% There is a binary even LCD $[34,10,12]$ code.
% There is a binary even LCD $[36,6,16]$ code.
% There is a binary even LCD $[40,6,18]$ code.
% There is a binary even LCD $[40,8,16]$ code.
\end{lem}

For $(n,k,d) \in {\mathcal P}_{2}$, 
we describe how
the above binary even LCD $[n,k,d]$ codes
$C_{2,34}$, $C_{2,36}$, $C_{2,40}$ and $C'_{2,40}$
were constructed, respectively.
As the first step, by using the method in~\cite[Section~6]{HS}, 
our computer search by {\sc Magma} found
a binary even LCD $[n,k,d-2]$ code
with generator matrix 
$\left( \begin{array}{cc}
I_k & A
\end{array}\right)$
for $(n,k,d) \in {\mathcal P}_{2}$ in~\eqref{eq:364040-2}.
Our feeling is that construction of 
a binary even LCD $[n,k,d-2]$ code is usually easier than
that of a binary even LCD $[n,k,d]$ code by the method.
As the next step,
by applying Theorem~\ref{thm:II} to $A$, 
our computer search by {\sc Magma} found
a binary even LCD $[n,k,d]$ code.
Note that the minimum weight is increased by $2$ in this case.
The codes $C_{2,34}$, $C_{2,36}$, $C_{2,40}$ and $C'_{2,40}$
have the following generator matrices:
\[
\left( \begin{array}{cc}
I_{10} & A_{2,34}
\end{array}\right),
\left( \begin{array}{cc}
I_6 & A_{2,36}
\end{array}\right),
\left( \begin{array}{cc}
I_6 & A_{2,40}
\end{array}\right) \text{ and }
\left( \begin{array}{cc}
I_8 & A'_{2,40}
\end{array}\right), 
\]
respectively, where
$A_{2,34}$, $A_{2,36}$, $A_{2,40}$ and $A'_{2,40}$
are listed in Figure~\ref{Fig:F2-2}.

%%%%%%%%%%%%%%%%%%%%%%%%%%
\begin{figure}[thbp]
\begin{center}
%{\small
{\footnotesize
%{\scriptsize
%{\tiny
\begin{align*}
B_{2,32}&=
\left(
\begin{array}{cccccccccccccccccccccccccccccccccccc}
1&0&0&1&1&1&0&1&0&0&1&1\\
0&1&0&1&1&1&1&1&1&1&1&0\\
1&0&0&0&0&0&0&1&1&1&0&1\\
1&0&0&1&1&0&1&1&1&0&1&0\\
0&1&0&1&0&1&0&1&0&1&1&1\\
0&1&0&0&0&0&0&1&1&1&1&0\\
1&0&0&0&1&0&0&0&1&0&1&1\\
1&0&0&0&1&0&1&1&0&0&0&1\\
1&0&0&0&1&0&1&0&0&1&1&0\\
0&1&0&0&1&1&1&0&1&0&1&1\\
1&0&0&1&0&1&0&0&1&1&0&0\\
0&1&0&0&1&0&0&0&0&1&1&1\\
0&1&0&0&0&1&1&0&0&1&0&1\\
1&0&0&1&0&1&1&1&0&0&0&0\\
1&0&0&1&0&0&0&0&0&1&1&1\\
0&1&0&1&1&0&0&1&0&0&0&1\\
0&1&0&1&0&0&1&0&0&0&1&1\\
0&1&0&1&0&0&0&0&1&1&0&1\\
0&1&0&0&1&1&0&1&0&0&1&0\\
0&1&0&0&1&0&1&1&0&1&0&0\\
\end{array}
\right)\\
B_{2,34}&=
\left(
\begin{array}{cccccccccccccccccccccccccccccccccccc}
1&0&1&0&1&0&1&0&0&0&0&1\\
1&0&0&1&0&0&1&1&0&1&1&1\\
1&0&1&1&0&1&0&1&0&1&1&0\\
1&0&1&1&0&0&0&0&0&1&0&1\\
1&0&1&1&1&0&0&1&0&0&1&1\\
1&0&0&0&0&1&1&1&1&0&0&0\\
1&0&0&1&0&0&0&0&1&1&1&0\\
1&0&0&1&1&1&1&0&0&1&1&0\\
1&0&0&1&1&0&1&1&1&0&1&0\\
1&0&1&1&1&0&1&0&1&1&0&0\\
1&0&0&0&1&1&0&1&1&1&0&1\\
1&0&1&0&0&0&0&1&1&1&0&0\\
1&0&0&1&1&0&0&1&0&1&0&0\\
1&0&0&0&0&0&1&0&1&1&0&1\\
1&0&0&1&1&1&1&1&0&0&0&1\\
1&0&0&0&1&0&0&0&1&0&1&1\\
1&0&1&0&0&1&0&0&0&0&1&1\\
1&0&1&1&0&0&1&0&0&0&1&0\\
1&0&1&0&1&1&1&1&1&0&1&1\\
1&0&1&1&0&1&0&0&1&0&0&0\\
1&0&1&0&1&1&0&1&0&0&0&0\\
0&0&1&1&1&1&1&0&0&0&0&0\\
\end{array}
\right)
\end{align*}
}
\end{center}
\caption{Matrices $B_{2,34}$ and $B_{2,34}$}
\label{Fig:F2-3}
\end{figure}
%%%%%%%%%%%%%%%%%%%%%%%%%%

%%%% construction I %%%%
Now, by the {\sc Magma} function {\tt BestKnownLinearCode},
one can construct 
a binary $[32, 20, 6]$ code $D_{2,32}$
and 
a binary $[34, 22, 6]$ code $D_{2,34}$.
The codes $D_{2,32}$ and $D_{2,34}$  have the following generator matrices:
\[
G_{2,32}=
\left( \begin{array}{cc}
I_{20} & B_{2,32}
\end{array}\right)
\text{ and }
G_{2,34}=
\left( \begin{array}{cc}
I_{22} & B_{2,34}
\end{array}\right), 
\]
respectively, 
where $B_{2,32}$ and $B_{2,34}$ are listed in Figure~\ref{Fig:F2-3}.
Using Lemma~\ref{lem:LCD},
we verified by {\sc Magma} that $D_{2,32}$ and $D_{2,34}$  are LCD\@. 
%%%
By applying Theorem~\ref{thm:I} to 
$G_{2,32}$ and $G_{2,34}$, 
our computer search by {\sc Magma} found
a binary LCD $[34,21,6]$ code $D'_{2,34}$ as $D_{2,32}(x)$ and 
a binary LCD $[36,23,6]$ code $D_{2,36}$ as $D_{2,34}(x')$,
where
\begin{align*}
x&=(0,\ldots,0,1,1,1,0,1,1,0,0,1,0,0,0) \text{ and }\\
x'&=(0,\ldots,0,1,1,0,1,1,0,0,1,1,1,1).
\end{align*}

\begin{lem}\label{lem:2-32-36}
For 
$
(n,k) \in \{
(32, 20),
(34,21),
(34, 22),
(36,23)
\}$,
there is a binary LCD $[n,k,6]$ code.
\end{lem}

\begin{lem}\label{lem:2-33}
There is a binary LCD $[33,22,6]$ code.
\end{lem}
\begin{proof}
The second column of $B_{2,34}$ is $\0_{22}^T$.
Let $D_{2,33}$ be the binary $[33,22,6]$ code constructed
from $D_{2,34}$ as the punctured code by deleting the $24$-th coordinate.
It is obvious that $D_{2,33}$ is LCD\@.
\end{proof}

Let $\hat{d}_2(n,k)$ denote the largest minimum weight among all
binary $[n,k]$ codes.
The current information on $\hat{d}_2(n,k)$ can be found in~\cite{Grassl}.
For example, it is known that
\begin{equation}\label{eq:364040}
\hat{d}_2(34,10)=12,
\hat{d}_2(36,6)=\hat{d}_2(40,8)=16
\text{ and }
\hat{d}_2(40,6)=18.
\end{equation}

\begin{lem}\label{lem:2-32-19}
There is a binary LCD $[32,19,6]$ code.
\end{lem}
\begin{proof}
By Lemma~\ref{lem:2-32-36}, there is a binary LCD $[32,20,6]$ code.
It follows from Lemma~\ref{lem:bound} with~\eqref{eq:364040} that
\[
6 \le d_2(32,20) \le d_2(32,19) \le \hat{d}_2(32, 19)=6.
\]
The result follows.
\end{proof}

%%%%%%%%% n >30 %%%%%%%%%%%%
%%%%%%%%% n >30 %%%%%%%%%%%%
%%%%%%%
\subsection{Determination of $d_2(n,k)$}
It is previously known~\cite{FLFR} that
% \begin{equation}\label{eq:binary}
% \begin{split}
% &d_2(n,k)\in \{7,8\}\text{ for $(n,k)=(28,13)$, $(29,13)$ and  $(30,14)$,}\\
% &d_2(n,k)\in \{11,12\} \text{ for $(n,k)=(28,6)$ and $(30,7)$.}
% \end{split}
% \end{equation}
\begin{equation}\label{eq:binary}
d_2(n,k) \in
\begin{cases}
\{7,8\}   &\text{ if $(n,k)=(28,13)$, $(29,13)$ and $(30,14)$,}\\
\{11,12\} &\text{ if $(n,k)=(28,6)$ and $(30,7)$.}
\end{cases}
\end{equation}
Also, it is known~\cite{Grassl} that
\begin{equation}\label{eq:d6}
%\hat{d}_2(32, 19)=
%\hat{d}_2(32, 20)=
%\hat{d}_2(34,21)=
%\hat{d}_2(34, 22)=
%\hat{d}_2(36,23)=
%6.
\begin{split}
\hat{d}_2(n, k)=6 \text{ if } (n,k) &=  
(32, 19), (32, 20), \\
& (33,22), (34,21), (34, 22) \text{ and } (36,23).
\end{split}
\end{equation}
From Lemmas~\ref{lem:2-1}--\ref{lem:2-32-19}
with~\eqref{eq:364040}--\eqref{eq:d6}, 
we have the following:

\begin{thm}
Let $d_2(n,k)$ denote the largest minimum weight among all
binary LCD $[n,k]$ codes.  Then
% \begin{align*}
% &
% d_2(32,19)=
% d_2(32, 20)=
% d_2(34,21)=
% d_2(34, 22)=
% d_2(36,23)=6,
% \\ &
% d_2(28,13) =d_2(29,13)=d_2(30,14)=8
% \\ &
% d_2(28,6) =d_2(30,7)=d_2(34,10)=12,
% \\ &
% d_2(36,6)=d_2(40,8)=16  \text{ and }
% d_2(40,6)=18.
% \end{align*}
\begin{align*}
d_2(n,k)=
\begin{cases}
6 & \text{ if } (n,k)=(32,19), (32,20), (33,22),  \\
   & \qquad \qquad \qquad (34,21), (34,22) \text{ and } (36,23), \\ 
8 & \text{ if } (n,k)=(28,13), (29,13) \text{ and }(30,14), \\
12 & \text{ if } (n,k)=(28,6), (30,7) \text{ and }(34,10),\\
16 & \text{ if } (n,k)=(36,6) \text{ and } (40,8),\\
18 & \text{ if } (n,k)=(40,6).
\end{cases}
\end{align*}
\end{thm}

If there is a binary LCD $[n,k,d]$ code with $k \ge 3$, then 
there  is a binary LCD $[n+(2^k-1)s,k,d+2^{k-1}s]$ code for every positive integer $s$~\cite[Lemma~3.5]{AHS2}.
Hence, as a consequence of Lemmas~\ref{lem:2-1}--\ref{lem:2-32-19},
we have the following:

\begin{cor}
Suppose that $(n,k,d) \in {\mathcal P}'_{2}$, where
\[
{\mathcal P}'_{2}=
\left\{
\begin{array}{l}
(28,13, 8),
(28, 6,12),
(29,13, 8),
(30,14, 8),
(30, 7,12),
\\
(32,19, 6),
(32,20, 6),
(33,22,6),
(34,10,12),
(34,21, 6),
\\
(34,22, 6),
(36, 6,16),
(36,23, 6),
(40, 6,18),
(40, 8,16)
\end{array}
\right\}.
\]
For a nonnegative integer $s$,
there is a binary LCD $[n+(2^k-1)s,k,d+2^{k-1}s]$ code.
\end{cor}

\section{New quaternary Hermitian LCD codes}\label{sec:4}

The current information on the largest minimum weights $d^H_4(n,k)$
can be found in~\cite{LZYC} for $n \le 25$ (see also~\cite{H}).
In this section, by Theorem~\ref{thm:I} we construct quaternary Hermitian LCD codes,
which improve the
previously known lower bounds on the largest minimum weights $d^H_4(n,k)$.
These new quaternary Hermitian LCD codes establish the new existence of 
some entanglement-assisted quantum codes.
All computer calculations in this section
were done with the help of {\sc Magma}~\cite{Magma}.

%%%%%%%
\subsection{Lengths $n=21$ and $22$}
There are quaternary Hermitian LCD codes with parameters
$[19,7,9]$ and $[20,7,10]$~\cite{H}.
We denote these codes by $C_{4,19}$ and $C_{4,20}$, respectively.
The code $C_{4,20}$ has generator matrix $G_{4,20}=
\left(
\begin{array}{cc}
I_7 & M_{20}
\end{array}
\right)
$, where $M_{20}$ is listed in~\cite[Fig.~1]{H}.
Since $C_{4,19}$ is constructed from $C_{4,20}$ as 
the punctured code by deleting the first coordinate~\cite{H},
$C_{4,19}$ has the following generator matrix:
\[
G_{4,19}=
\left(
\begin{array}{cccccc}
0 & \cdots & 0 & &&\\
  &  I_6 &   & &M_{20}&\\
\end{array}
\right).
\]
By applying Theorem~\ref{thm:I} to the generator matrices
$G_{4,19}$ and $G_{4,20}$ of the codes $C_{4,19}$ and $C_{4,20}$,
it is possible to construct
quaternary Hermitian LCD $[21,8]$ codes and
quaternary Hermitian LCD $[22,8]$ codes, respectively.
In particular,
our computer search by {\sc Magma} found
a quaternary Hermitian LCD $[21,8,9]$ code $C_{4,21}$ as $C_{4,19}(x)$
and
a quaternary Hermitian LCD $[22,8,10]$ code $C_{4,22}$ as $C_{4,20}(x')$,
% for parameters $(n,k,d)=(22,8,10)$ (resp.\ $(21,8,9)$).
where 
\begin{align*}
x&=(0,\ldots,0,\ww,\ww,0,0,\ww,0,\ww,\vv,\ww,1,\vv)  \text{ and }\\
x'&=(0,\ldots,0,\ww,1,1,\ww,0,\ww,0,1,0,\vv,1,\ww,\vv).
\end{align*}
% The minimum weights of $C_{4,22}$ and $C_{4,21}$
% were calculated by {\sc Magma}.
% The codes $C_{4,21}$ and $C_{4,22}$ have the following
% weight enumerators
% \begin{align*}
% & 1+93 y^{9}+387 y^{10}+912 y^{11}+2262 y^{12}+4737 y^{13}
% +8382 y^{14} +11724 y^{15}
% \\&
% +13047 y^{16}+11379 y^{17}+7719 y^{18}+3636 y^{19}
% +1098 y^{20}+159  y^{21},
% \\&
% 1+177 y^{10}+504 y^{11}+1341 y^{12}+2832 y^{13}
% +5661 y^{14}+9276 y^{15} +11925 y^{16}
% \\&
% +12492 y^{17}+10683 y^{18}+6684 y^{19}+2997 y^{20}+852 y^{21}+111 y^{22},
% \end{align*}
% respectively.

\begin{lem}\label{lem:4-1}
There is a  quaternary Hermitian LCD $[22,8,10]$ code.
There is a  quaternary Hermitian LCD $[21,8,9]$ code.
\end{lem}

By the {\sc Magma} function {\tt BestKnownLinearCode},
one can construct a quaternary $[19,9,8]$ code $D_{4,19}$.
Using Lemma~\ref{lem:LCD},
we verified by {\sc Magma} that $D_{4,19}$ is Hermitian LCD\@. 
The code $D_{4,19}$ has generator matrix 
$G_{4,19}=
\left( \begin{array}{cc}
I_9 & A_{4,19}
\end{array}\right)$,
where $ A_{4,19}$ is listed in Figure~\ref{Fig:F4}.
%%%%%%%%%%%%%%%%
By using the method in~\cite[Section~2]{H},
our computer search by {\sc Magma} found
a quaternary Hermitian LCD $[19,10,7]$ code $D'_{4,19}$
with generator matrix
$G'_{4,19}=\left( \begin{array}{cc}
I_{10} & A'_{4,19}
\end{array}\right)$, where $A'_{4,19}$ is listed in Figure~\ref{Fig:F4}.
% The minimum weight was calculated by {\sc Magma}.
% Using Lemma~\ref{lem:LCD}, it was verified by {\sc Magma}
% that  $D'_{4,19}$ is Hermitian LCD\@.
%%%%%%%%%%%%%%%%
By applying Theorem~\ref{thm:I} to the generator matrices 
$G_{4,19}$ and $G'_{4,19}$
of the codes $D_{4,19}$ and $D'_{4,19}$,
it is possible to construct
quaternary Hermitian LCD $[21,10]$ codes and
quaternary Hermitian LCD $[21,11]$ codes, respectively.
In particular,
our computer search by {\sc Magma} found
a quaternary Hermitian LCD $[21,10,8]$ code $D_{4,21}$ as $D_{4,19}(x)$
and a quaternary Hermitian LCD $[21,11,7]$ code $D'_{4,21}$ as $D'_{4,19}(x')$,
where 
\begin{align*}
x&=(0,\ldots ,0 ,1,\ww,0,\ww,\ww,\vv,\vv,0,\vv,1)  \text{ and }\\
x&'=(0,\ldots ,0 ,\ww,0,\ww,1,0,\ww,\vv,0,1).
\end{align*}
% The minimum weights of $D_{4,21}$ and $D'_{4,21}$
% were calculated by {\sc Magma}.

%Generator matrices $G_{4,21}$ and $G'_{4,21}$ of $D_{4,21}$ and
%$D'_{4,21}$ are listed in Figure~\ref{Fig:F4}.
% The codes $D_{4,21}$ and $D'_{4,21}$ have the following
% weight enumerators
% \begin{align*}
% &1+67 y^{8}+1896 y^{9}+5172 y^{10}+13125 y^{11}+37791 y^{12}+79197 y^{13}
% \\&
% +131304 y^{14}+185802 y^{15}+208980 y^{16}+183114 y^{17}+124116 y^{18} 
% \\&
% +58065 y^{19}+17169 y^{20}+2577 y^{21},
% \\&
% 1+ 303 y^{7}+ 1380 y^{8}+ 5451 y^{9}+19605 y^{10}+59268 y^{11}+149565 y^{12}
% \\&
% +310566 y^{13}+529314 y^{14}+741351 y^{15}+835665 y^{16}+738687 y^{17} 
% \\&
% +490809 y^{18}+232470 y^{19}+69789 y^{20}+10080 y^{21},
% \end{align*}
% respectively.

\begin{lem}\label{lem:4-2}
There is a  quaternary Hermitian LCD $[21,10,8]$ code.
There is a  quaternary Hermitian LCD $[21,11,7]$ code.
\end{lem}

%%%%%%%%%%%%%%%%%%%%%%%%%%%%%%%%%%%%%%%%%%%%%%%%%%%%
\begin{figure}[thb]
\begin{center}
%{\small
{\footnotesize
%{\scriptsize
%{\tiny
\begin{align*}
A_{4,19}=&
\left(
\begin{array}{cccccccccccccccccccccccccccccccc}
% 1&0&0&0&0&0&0&0&0&0&0&1&\ww&0&\ww&\ww&\vv&\vv&0&\vv&1\\
% \vv&\vv&1&0&0&0&0&0&0&0&0&1&\vv&\ww&\ww&0&1&0&\vv&\vv&\ww\\
% \ww&\ww&0&1&0&0&0&0&0&0&0&\ww&0&0&1&\ww&\ww&1&1&\ww&0\\
% \ww&\ww&0&0&1&0&0&0&0&0&0&0&\ww&0&0&1&\ww&\ww&1&1&\ww\\
% \vv&\vv&0&0&0&1&0&0&0&0&0&\ww&1&1&\vv&0&\vv&\ww&\vv&0&\ww\\
% 1&1&0&0&0&0&1&0&0&0&0&\ww&\vv&\ww&\ww&\vv&\ww&\vv&\vv&\ww&\vv\\
% \ww&\ww&0&0&0&0&0&1&0&0&0&\vv&0&\ww&\vv&\ww&0&\ww&1&1&\vv\\
% \vv&\vv&0&0&0&0&0&0&1&0&0&\vv&1&1&\vv&\vv&1&0&0&\vv&0\\
% \vv&\vv&0&0&0&0&0&0&0&1&0&0&\vv&1&1&\vv&\vv&1&0&0&\vv\\
% \ww&\ww&0&0&0&0&0&0&0&0&1&\vv&\ww&\ww&0&1&0&\vv&\vv&\ww&1
1&\vv&\ww&\ww&0&1&0&\vv&\vv&\ww\\
\ww&0&0&1&\ww&\ww&1&1&\ww&0\\
0&\ww&0&0&1&\ww&\ww&1&1&\ww\\
\ww&1&1&\vv&0&\vv&\ww&\vv&0&\ww\\
\ww&\vv&\ww&\ww&\vv&\ww&\vv&\vv&\ww&\vv\\
\vv&0&\ww&\vv&\ww&0&\ww&1&1&\vv\\
\vv&1&1&\vv&\vv&1&0&0&\vv&0\\
0&\vv&1&1&\vv&\vv&1&0&0&\vv\\
\vv&\ww&\ww&0&1&0&\vv&\vv&\ww&1
\end{array}
\right)
\\
A'_{4,19}=&
\left(
\begin{array}{ccccccccccccccccccccccccccccccc}
% 1&0&0&0&0&0&0&0&0&0&0&0&\ww&0&\ww&1&0&\ww&\vv&0&1\\
% 1&1&1&0&0&0&0&0&0&0&0&0&0&0&0&1&1&1&1&1&1\\
% \ww&\ww&0&1&0&0&0&0&0&0&0&0&1&\vv&1&1&1&1&0&0&0\\
% \ww&\ww&0&0&1&0&0&0&0&0&0&0&1&\ww&\ww&\ww&1&0&1&0&0\\
% \ww&\ww&0&0&0&1&0&0&0&0&0&0&1&1&\vv&\vv&1&0&0&1&0\\
% \vv&\vv&0&0&0&0&1&0&0&0&0&0&1&1&\ww&0&0&1&1&1&0\\
% \vv&\vv&0&0&0&0&0&1&0&0&0&0&1&\vv&\vv&\ww&\vv&\vv&1&1&0\\
% \ww&\ww&0&0&0&0&0&0&1&0&0&0&1&1&0&\vv&\vv&1&0&0&1\\
% \ww&\ww&0&0&0&0&0&0&0&1&0&0&1&0&1&0&1&\ww&1&0&1\\
% \ww&\ww&0&0&0&0&0&0&0&0&1&0&1&0&\vv&\ww&\vv&1&\ww&\vv&1\\
% 1&1&0&0&0&0&0&0&0&0&0&1&1&\ww&\vv&\vv&\ww&\ww&\vv&1&\vv
% 
0&0&0&1&1&1&1&1&1\\
1&\vv&1&1&1&1&0&0&0\\
1&\ww&\ww&\ww&1&0&1&0&0\\
1&1&\vv&\vv&1&0&0&1&0\\
1&1&\ww&0&0&1&1&1&0\\
1&\vv&\vv&\ww&\vv&\vv&1&1&0\\
1&1&0&\vv&\vv&1&0&0&1\\
1&0&1&0&1&\ww&1&0&1\\
1&0&\vv&\ww&\vv&1&\ww&\vv&1\\
1&\ww&\vv&\vv&\ww&\ww&\vv&1&\vv
\end{array}
\right)
\end{align*}
}
\end{center}
\caption{Matrices $A_{4,19}$ and $A'_{4,19}$}
\label{Fig:F4}
\end{figure}
%%%%%%%%%%%%%%%%%%%%%%%%%%%%%%%%%%%%%%%%%%%%%%%%%%%%

%%%%%%%%% n >20 %%%%%%%%%%%%
%%%%%%%%% n >20 %%%%%%%%%%%%
%%%%%%%
\subsection{Lengths $n$ with $23 \le n \le 30$}
Now let us look at construction of quaternary Hermitian LCD $[n,k]$ codes for
$23 \le n \le 30$.
Note that no information on $d^H_4(n,k)$ is known for $n \ge 26$.

Suppose that $(n,k,d) \in {\mathcal P}_{4}$, where
\begin{equation}\label{eq:Gamma4}
{\mathcal P}_{4}=\{
( 23, 18, 4), 
( 24, 16, 6), 
( 26, 21, 4), 
( 27, 22, 4), 
( 28, 23, 4)
\}.
\end{equation}
By the {\sc Magma} function {\tt BestKnownLinearCode},
one can construct a quaternary $[n,k,d]$ code.
By considering an equivalent code, 
we have a quaternary $[n,k,d]$ code
$C_{4,n}$ with the following generator matrix:
\[
G_{4,n}=
\left( \begin{array}{cc}
I_{k} & A_{4,n}
\end{array}\right),
\]
where
$A_{4,n}$ is listed in Figure~\ref{Fig:F4-2}.
Using Lemma~\ref{lem:LCD},
we verified by {\sc Magma} that $C_{4,n}$ is Hermitian LCD\@. 
By applying Theorem~\ref{thm:I} to $G_{4,n}$,
our computer search by {\sc Magma} found a
quaternary Hermitian LCD $[n+2,k+1,d]$ code
$C'_{4,n+2}$ as $C_{4,n}(x_n)$,
where
\[
\begin{array}{ll}
x_{23}=(0,\ldots,0,\vv,0,\ww,\ww,\ww),&
x_{24}=(0,\ldots,0,1,\ww,\vv,\vv,\ww,\vv),\\
x_{26}=(0,\ldots,0,\ww,0,1,1,1),&
x_{27}=(0,\ldots,0,\vv,\vv,0,\ww,1) \text{ and }\\
x_{28}=(0,\ldots,0,\vv,0,1,\ww,1).
\end{array}
\]

\begin{lem}\label{lem:4-3}
For $(n,k,d) \in {\mathcal P}_{4}$ in~\eqref{eq:Gamma4},
there is a quaternary Hermitian LCD $[n,k,d]$ code and
there is a quaternary Hermitian LCD $[n+2,k+1,d]$ code.
\end{lem}

%%%%%%%%%%%%%%%%%%%%%%%%%%%%%%%%%%%%%%%
\begin{figure}[thbp]
\begin{center}
%{\small
{\footnotesize
%{\scriptsize
%{\tiny
\begin{align*}
A_{4,23}&=
\left(\begin{array}{ccccccccccccccccc}
1&1&1&\vv&\vv\\
\ww&1&\vv&0&\ww\\
\vv&1&0&0&\vv\\
0&0&\vv&\ww&\ww\\
0&\vv&\ww&\vv&\vv\\
1&0&0&\ww&\ww\\
\ww&\ww&0&1&\ww\\
\ww&\vv&0&\vv&1\\
\ww&1&1&\vv&0\\
\ww&\vv&1&1&\ww\\
1&1&1&\ww&\ww\\
1&0&\vv&1&1\\
1&\vv&\ww&\ww&\ww\\
0&1&\vv&\vv&\ww\\
\ww&\ww&\vv&\ww&1\\
\vv&1&1&0&\ww\\
0&\ww&\vv&0&1\\
0&\ww&\ww&\ww&\vv\\
\end{array}\right)
A_{4,24}=
\left(\begin{array}{ccccccccccccccccc}
\ww&\vv&1&\ww&0&0&\ww&\ww\\
\ww&0&\ww&\vv&\ww&\ww&\vv&\vv\\
\ww&\vv&\ww&\ww&\vv&1&\vv&1\\
1&\vv&\vv&\vv&\vv&0&1&\ww\\
0&\ww&\vv&0&1&\vv&\vv&1\\
\vv&\vv&\ww&0&0&\ww&0&\ww\\
\vv&1&1&0&\vv&0&0&\vv\\
0&\ww&\ww&\vv&\vv&0&0&\vv\\
1&1&\ww&\vv&\ww&1&\ww&1\\
\vv&\ww&\ww&1&1&1&\ww&\ww\\
0&\ww&1&\ww&\ww&1&\ww&1\\
\ww&0&\vv&\ww&0&\ww&1&\vv\\
\ww&0&0&\vv&\ww&\vv&1&0\\
1&\ww&\vv&1&\vv&\vv&0&0\\
0&\ww&\vv&0&\vv&1&\ww&0\\
\vv&1&\vv&0&1&\vv&\ww&0\\
\end{array}\right)
\\
A_{4,26}&=
\left(\begin{array}{ccccccccccccccccc}
\vv&1&\ww&1&\vv\\
0&1&\ww&\vv&\vv\\
0&\ww&\ww&0&\vv\\
\vv&1&0&\ww&\ww\\
\ww&0&1&\vv&0\\
\vv&\vv&0&1&\ww\\
\ww&\vv&\vv&\vv&0\\
\vv&\ww&1&0&\ww\\
0&1&\ww&\ww&\ww\\
\vv&0&1&\vv&1\\
\ww&0&\ww&1&\vv\\
1&\vv&0&\vv&0\\
\ww&0&1&\ww&\ww\\
1&\vv&\ww&1&0\\
0&\vv&\vv&1&\ww\\
\ww&\ww&\ww&\vv&0\\
0&1&\vv&\vv&1\\
1&1&0&1&\ww\\
0&\ww&\vv&0&\ww\\
\vv&\ww&\ww&\vv&\ww\\
\ww&1&0&\vv&\vv\\
\end{array}\right)
A_{4,27}=
\left(\begin{array}{ccccccccccccccccc}
\vv&\ww&0&0&1\\
0&\ww&\ww&0&1\\
0&\ww&\ww&\ww&\ww\\
\vv&0&\vv&1&1\\
\ww&1&\vv&0&0\\
\vv&0&\ww&0&\vv\\
\ww&\vv&0&\ww&\vv\\
\vv&1&\ww&\vv&\ww\\
0&\ww&\vv&1&1\\
\vv&1&\vv&1&0\\
\ww&\ww&1&\vv&0\\
1&0&0&\ww&\vv\\
\ww&1&\ww&\ww&0\\
1&\ww&\ww&\ww&\vv\\
0&\vv&\ww&0&\vv\\
\ww&\ww&1&1&\ww\\
0&\vv&\ww&\ww&1\\
1&0&0&1&1\\
0&\vv&\ww&\vv&\ww\\
\vv&\ww&1&\vv&\ww\\
\ww&0&\ww&0&1\\
0&0&1&\vv&1\\
\end{array}\right)
A_{4,28}=
\left(\begin{array}{ccccccccccccccccc}
1&1&1&\ww&1\\
\ww&\ww&0&1&0\\
\ww&\ww&\ww&\ww&0\\
\vv&\ww&0&\ww&1\\
\vv&0&\vv&\ww&\vv\\
\vv&\vv&1&0&1\\
1&\vv&1&1&\vv\\
\ww&\vv&\ww&1&1\\
\ww&\vv&1&1&0\\
\ww&\ww&0&\vv&1\\
0&\ww&0&\ww&\vv\\
1&\ww&0&\ww&\ww\\
\vv&1&1&\ww&\vv\\
\vv&0&0&\ww&\ww\\
\vv&\ww&0&\vv&0\\
0&\ww&\ww&0&\vv\\
\vv&\ww&\ww&1&0\\
1&\ww&\vv&0&\ww\\
\vv&\ww&\vv&\ww&0\\
1&0&\ww&1&1\\
\ww&1&\vv&\vv&\vv\\
1&\ww&\ww&1&\ww\\
0&1&\vv&1&0\\
\end{array}\right)
\end{align*}
}
\end{center}
\caption{Matrices $A_{4,23}$, $A_{4,24}$, $A_{4,26}$, $A_{4,27}$ and $A_{4,28}$}
\label{Fig:F4-2}
\end{figure}
%%%%%%%%%%%%%%%%%%%%%%%%%%%%%%%%%%%%%%%%%%%%%%%%%%%%

%%%%%%%%% n >20 %%%%%%%%%%%%
%%%%%%%%% n >20 %%%%%%%%%%%%

%%%%%%%
\subsection{Determination of $d^H_4(n,k)$ and bounds on $d^H_4(n,k)$}
It is previously known~\cite{LZYC} that
\begin{equation}\label{eq:F4-1}
\begin{split}
&
d^H_4(21,8) \in \{8,9,10\},
d^H_4(21,10) \in \{7,8,9\},
\\&
d^H_4(21,11) \in \{6,7,8\} \text{ and }
d^H_4(22,8) \in \{9,10,11\}.
\end{split}
\end{equation}
Let $\hat{d}_4(n,k)$ denote the largest minimum weight among all
quaternary $[n,k]$ codes.
The current information on $\hat{d}_4(n,k)$ can be found in~\cite{Grassl}.
For example, it is known that
\begin{equation}\label{eq:F4-2}
% \begin{split}
% &
% \hat{d}_4(23, 18)=\hat{d}_4(26, 21)=\hat{d}_4(27, 22)
% =\hat{d}_4(28, 23)=\hat{d}_4(30, 24)=4,
% \\&
% \hat{d}_4(24, 16)=6,
% \hat{d}_4(25, 19)\in\{4,5\},
% \hat{d}_4(26, 17)\in\{6,7\},
% \\&
% \hat{d}_4(28, 22)\in\{4,5\},
% \hat{d}_4(29, 23)\in\{4,5\}.
% \end{split}
\begin{split}
&
\hat{d}_4(n,k)=
\begin{cases}
4 & \text{ if } (n,k)=
(23, 18),(26, 21),\\ & \qquad \qquad \qquad (27, 22),  (28, 23) \text{ and }(30, 24),\\
6 & \text{ if } (n,k)=(24, 16),
\end{cases}
\\ & \text{and } \\ &
\hat{d}_4(n,k) \in 
\begin{cases}
\{4,5\} &\text{ if } (n,k)=(25, 19), (28, 22) \text{ and } (29, 23),\\
\{6,7\} &\text{ if } (n,k)=(26, 17).
\end{cases}
\end{split}
\end{equation}
From Lemmas~\ref{lem:4-1}--\ref{lem:4-3} 
with~\eqref{eq:F4-1} and \eqref{eq:F4-2}, we have the following:

\begin{thm}\label{thm:F4}
Let $d^H_4(n,k)$ denote the largest minimum weight among all
quaternary Hermitian LCD $[n,k]$ codes.
Then 
% \begin{enumerate} \item
% \begin{align*}
% &
% d^H_4(23, 18)=d^H_4(26, 21)=d^H_4(27, 22)=d^H_4(28, 23)=d^H_4(30, 24)=4,\\&
% d^H_4(24, 16)=6,d^H_4(22,8) =10 \text{ and }
% %\end{align*}
% %\item
% %\begin{align*}
% \\&
% d^H_4(21,8) \in \{9,10\},d^H_4(21,10) \in \{8,9\},d^H_4(21,11) \in \{7,8\},\\& 
% d^H_4(25, 19)\in\{4,5\},d^H_4(26, 17)\in\{6,7\},d^H_4(28, 22)\in\{4,5\},\\&
% d^H_4(29, 23)\in\{4,5\}.
% \end{align*}
\begin{align*}
&d^H_4(n,k)=
\begin{cases}
4 &\text{ if } (n,k)=
(23, 18),(26, 21),(27, 22),(28, 23) \text{ and } (30, 24),\\
6 &\text{ if } (n,k)=(24, 16),
\end{cases} \\
&\text{and }\\
%\end{align*}
%\item
%\begin{align*}
&d^H_4(n,k) \in
\begin{cases}
\{4,5\} &\text{ if } (n,k)=(25, 19), (28, 22) \text{ and } (29, 23), \\
\{6,7\} &\text{ if } (n,k)=(26, 17), \\
\{7,8\} &\text{ if } (n,k)=(21,11),\\
\{8,9\} &\text{ if } (n,k)=(21,10),\\
\{9,10\}&\text{ if } (n,k)=(21,8),\\
\{10,11\} &\text{ if } (n,k)=(22,8).
\end{cases}
\end{align*}
\end{thm}

\begin{cor}
\begin{align*}
d^H_4(n,k) \in 
\begin{cases}
\{4,5\} &\text{ if } (n,k)=
(26, 20) \text{ and } (27, 21),  \\
\{4,5,6\} &\text{ if } (n,k)=
(25, 18), (28, 21), (29, 22) \text{ and } (30, 23), \\
\{6,7,8\} &\text{ if } (n,k)=(26, 16).
\end{cases}
\end{align*}
\end{cor}
\begin{proof}
It is known~\cite{Grassl} that
$\hat{d}_4(n,k) \le 5$ if $(n,k)=(26, 20)$ and $(27, 21)$, 
$\hat{d}_4(n,k) \le 6$ if $(n,k)=(25, 18)$, $(28, 21)$, $(29, 22)$ and  $(30, 23)$, and
$\hat{d}_4(26,16) \le 8$.
The result follows from Lemma~\ref{lem:bound} and Theorem~\ref{thm:F4}.
\end{proof}

If there is a quaternary Hermitian LCD $[n,k,d]$ code with $k \ge 2$, then 
there  is a quaternary Hermitian LCD $[n+{\frac{4^k-1}{3}} s,k,
d+4^{k-1}s]$ code for every positive integer $s$~\cite[Lemma~3.3]{AH}.
Hence, 
%as a consequence of Lemmas~\ref{lem:4-1}--\ref{lem:4-3},
we have the following:

\begin{cor}\label{cor:4}
Suppose that $(n,k,d) \in {\mathcal P}'_{4}$, where
\begin{equation}\label{eq:4}
{\mathcal P}'_{4}=
\left\{
\begin{array}{l}
(21, 8,9),
(21,10,8),
(21,11,7),
(22,8,10),
(23,18,4),
\\
(24,16,6),
(25, 18,4), 
(25,19,4),
(26, 16,6),
(26,17,6),
\\
(26, 20,4),
(26,21,4),
(27, 21,4),
(27,22,4),
(28, 21,4), 
\\
(28,22,4),
(28,23,4),
(29, 22,4),
(29,23,4),
(30, 23,4),
\\
(30,24,4)
\end{array}
\right\}.
\end{equation}
For a nonnegative integer $s$,
there is a quaternary Hermitian LCD 
$[n+{\frac{4^k-1}{3}} s,k,d+4^{k-1}s]$ code.
%%%%%%%%%%%%%%%%%%%%5
% For a nonnegative integer $s$,
% there is a quaternary Hermitian LCD $[n,k,d]$ code, where
% \begin{equation}\label{eq:4}
% (n,k,d) \in 
% \left\{
% \begin{array}{l}
% (21+21845s,8,9+16384s),    \\
% (21+349525s,10,8+262144s),  \\
% (21+1398101s,11,7+1048576s), \\
% (22+21845s,8,10+16384s)
% \end{array}
% \right\}.
% \end{equation}
%%%%%%%%%%%%%%%%
% \begin{equation}\label{eq:4}
% \begin{split}
% (n,k,d)=&
%  (21+21845s,8,9+16384s),
% \\&
% (21+349525s,10,8+262144s),
% \\&
% (21+1398101s,11,7+1048576s),
% \\&
% (22+21845s,8,10+16384s).
% \end{split}
% \end{equation}
\end{cor}

\subsection{Entanglement-assisted quantum codes} 
An entanglement-assisted quantum
$[[n,k,d;c]]$ code $\cC$
encodes $k$ information qubits into $n$ channel qubits
with the help of $c$ pairs of maximally entangled Bell states.
The parameter $d$ is called the minimum weight of $\cC$.
% (ebits).
% Assuming the ebits of the receiver are error-free, it
The entanglement-assisted quantum code $\cC$
can correct up to $\lfloor \frac{d-1}{2} \rfloor$
errors acting on the $n$ channel qubits (see e.g.\ \cite{LLG}
and~\cite{LLGF}).
An entanglement-assisted quantum
$[[n,k,d;0]]$ code is a standard quantum code.
% An entanglement-assisted quantum
% $[[n,k,d;c]]$ code is called {\em optimal}
% if there is no entanglement-assisted quantum
% $[[n,k,d';c]]$ code for all $d' > d$.
% (see~\cite{GL} and~\cite{LBW}).
%R2 An entanglement-assisted quantum
%R2 $[[n,k,d;n-k]]_2$ code is called {\em maximal entanglement}.
% Some maximal entanglement
% entanglement-assisted quantum
% $[[n,k,d;c]]$ codes
% have better performance than all
% standard quantum $[[n+c,k,d]]$ codes
% (see e.g.~\cite{LLG} and~\cite{LLG17}).
%\cite{LB} and~\cite{LBW2}.
If there is a quaternary Hermitian LCD $[n,k,d]$ code,
then there is an
%R2 maximal entanglement
entanglement-assisted quantum
$[[n,k,d;n-k]]$ code
(see e.g.\ \cite{LLG} and~\cite{LLGF}).
Hence, as a consequence of Corollary~\ref{cor:4},
we have the following:

\begin{cor}
Suppose that $(n,k,d) \in {\mathcal P}'_{4}$, where ${\mathcal P}'_{4}$ is listed in~\eqref{eq:4}.
For a nonnegative integer $s$,
there is an entanglement-assisted quantum
$[[n+{\frac{4^k-1}{3}} s,k,d+4^{k-1}s;n+{\frac{4^k-1}{3}} s-k]]$ code.
\end{cor}

%%%%%%%%%%%%%%%%%%%%%%%%%%%%%%%%%%%%%%%%%%
\section{New ternary LCD codes}\label{sec:3}

A classification of ternary LCD codes was done in \cite{AH-C}
for $n \in \{1,2,\ldots,10\}$.
The largest minimum weights $d_3(n,k)$ were determined
for $n \in \{11,12,\ldots,19\}$ \cite[Table~7]{AHS}.
%% In addition, 
%% the largest minimum weights $d_3(20,k)$ with $4$ exceptions
%% were also determined, where
%% $d_3(20,7)\in \{8,9\}$,
%% $d_3(20,9)\in \{7,8\}$,
%% $d_3(20,12)\in \{5,6\}$ and
%% $d_3(20,15)\in \{3,4\}$.
% We remark that 
% very recently $d_3(20,12)$ and $d_3(20,15)$ have been determined
% in~\cite{LLFXM} as $6$ and $3$, respectively. 
%% Although
%% our extensive search failed to discover a ternary LCD $[20,k]$ code,
%% which determines $d_3(20,k)$ for $k \in \{7,9,12,15\}$,
In this section,
we give two examples of ternary LCD codes constructed by Theorem~\ref{thm:I},
which have minimum weights meeting the lower bounds on the largest
minimum weights among currently known all ternary codes.
%Theorem~\ref{thm:I} seems to be less useful for ternary LCD codes.
All computer calculations in this section
were done with the help of {\sc Magma}~\cite{Magma}.

Let $\hat{d}_3(n,k)$ denote the largest minimum weight among all
ternary $[n,k]$ codes.
The current information on $\hat{d}_3(n,k)$ can be found in~\cite{Grassl}.
For example, it is known that
\begin{equation}\label{eq:F3}
\begin{split}
&
\hat{d}_3(34, 22) \in \{7,8\},
\hat{d}_3(37, 23) \in \{7,8,9\},
\\&
\hat{d}_3(37, 29) = 5 \text{ and }
\hat{d}_3(40, 30) \in  \{5,6\}.
\end{split}
\end{equation}
By the {\sc Magma} function {\tt BestKnownLinearCode},
one can construct
a ternary $[34, 22, 7]$ code $C_{3,34}$
and
a ternary $[37, 29, 5]$ code $C_{3,37}$.
The codes  $C_{3,34}$ and   $C_{3,37}$ have the following generator matrices:
\[
G_{3,34}=
\left( \begin{array}{cc}
I_{22} & A_{3,34}
\end{array}\right)
\text{ and }
G_{3,37}=
\left( \begin{array}{cc}
I_{29} & A_{3,37}
\end{array}\right),
\]
respectively,
where $A_{3,34}$ and $A_{3,37}$ are listed in 
Figure~\ref{Fig:F3}.
Using Lemma~\ref{lem:LCD},
we verified by {\sc Magma} that  $C_{3,34}$ and   $C_{3,37}$  are LCD\@.
By applying Theorem~\ref{thm:I} to
$G_{3,34}$ and $G_{3,37}$,
our computer search by {\sc Magma} found
a ternary LCD $[37, 23, 7]$ code $C'_{3,37}$ as $C_{3,34}(x,(1,0,0))$ and
a ternary LCD $[40, 30, 5]$ code $C_{3,40}$ as $C_{3,37}(x',(1,0,0))$,
where
\begin{align*}
x&=(0,\ldots,0,2,1,1,1,2,1,1,0,0,0,0)  \text{ and }\\
x'&=(0,\ldots,0,1,1,1,1,0,0,0,0).
\end{align*}
%[ 34, 22, 7 ] 7-8
%[ 37, 23, 7 ] 7-9
%[ 37, 29, 5 ] 5
%[ 40, 30, 5 ] 5-6

\begin{lem}\label{lem:F3}
For
$(n,k,d) \in {\mathcal P}_3$, where 
\begin{equation}\label{eq:F3-2}
{\mathcal P}_3=
  \{
(34, 22, 7),
(37, 23, 7),
(37, 29, 5),
(40, 30, 5)
\},
\end{equation}
there is a ternary LCD $[n,k,d]$ code.
\end{lem}

From Lemma~\ref{lem:F3} with~\eqref{eq:F3}, we have the following:

\begin{prop}\label{prop:F3}
Let $d_3(n,k)$ denote the largest minimum weight among all
ternary LCD $[n,k]$ codes.  Then
\[
\begin{split}
&
{d}_3(34, 22) \in \{7,8\},
{d}_3(37, 23) \in \{7,8,9\},
\\&
{d}_3(37, 29) = 5 \text{ and }
{d}_3(40, 30) \in  \{5,6\}.
\end{split}
\]
\end{prop}

\begin{cor}
\[
\begin{split}
&
{d}_3(34, 21) \in \{7,8,9\},
{d}_3(37, 22) \in \{7,8,9\},
\\&
{d}_3(37, 28) \in \{5,6\} \text{ and }
{d}_3(40, 29) \in  \{5,6,7\}.
\end{split}
\]
\end{cor}
\begin{proof}
It is known~\cite{Grassl} that
$\hat{d}_3(34, 21) \le 9$,
$\hat{d}_3(37, 22) \le 9$,
$\hat{d}_3(37, 28) \le 6$ and
$\hat{d}_3(40, 29) \le 7$.
The result follows from Lemma~\ref{lem:bound} and Proposition~\ref{prop:F3}.
\end{proof}

If there is a ternary LCD $[n,k,d]$ code with $k \ge 2$, then 
there  is a ternary LCD $[n+{\frac{3^k-1}{2}} s,k,
d+3^{k-1}s]$ code for every positive integer $s$~\cite[Lemma~3.5]{AHS2}.
Hence, we have the following: 

\begin{cor}
Suppose that $(n,k,d) \in {\mathcal P}_{3} \cup {\mathcal P}'_3$, where
${\mathcal P}_{3}$ is listed in~\eqref{eq:F3-2} and 
\[
{\mathcal P}'_{3}=
\{
(34, 21,7),
(37, 22,7),
(37, 28,5),
(40, 29,5)
\}.
\]
For a nonnegative integer $s$,
there  is a ternary LCD $[n+{\frac{3^k-1}{2}} s,k,
d+3^{k-1}s]$ code.
\end{cor}

%%%%%%%%%%%%%%%%%%%%%%%%%%%%%%%%%%%
\bigskip
\noindent
{\bf Acknowledgments.}
This work was supported by JSPS KAKENHI Grant Number 19H01802.

%%%%%%%%%%%%%%%%%%%  References  %%%%%%%%%%%%%%%%%%%%%%%%

\begin{landscape}
%%%%%%%%%%%%%%%%%%%%%%%%%%%%%%%%%%%%%%%%%%%%%%%%%%%%
\begin{figure}[thbp]
\begin{center}
%{\small
{\footnotesize
%{\scriptsize
%{\tiny
\begin{align*}
A_{2,34}=&
\left(
\begin{array}{cccccccccccccccccccccccccccccccccccc}
0&1&0&0&0&1&1&0&0&1&1&0&1&1&1&1&1&1&1&1&1&0&0&1\\
1&1&1&0&0&1&1&0&1&0&0&1&1&1&1&1&0&1&1&1&0&1&1&1\\
1&1&0&1&0&1&1&0&1&1&1&0&0&0&0&0&0&1&1&0&1&0&0&0\\
1&0&0&0&0&0&1&1&0&1&0&0&0&0&0&1&1&1&0&1&1&0&1&1\\
0&1&1&0&0&1&1&1&1&1&0&0&1&0&0&0&0&1&0&0&0&0&1&1\\
0&0&0&0&1&1&1&0&1&0&0&1&0&1&0&1&0&1&1&0&1&0&0&1\\
1&1&0&0&1&1&0&1&0&1&1&1&1&1&1&1&1&1&0&1&0&1&1&0\\
0&0&0&1&0&0&1&1&0&1&1&1&0&1&0&0&0&0&1&1&1&0&0&1\\
1&0&1&1&1&1&1&0&0&1&0&0&1&0&0&0&1&1&0&1&0&1&0&1\\
1&0&0&1&1&1&1&1&1&0&0&0&1&0&1&1&0&0&1&0&0&0&0&0\\
\end{array}
\right)
\\
A_{2,36}=&
\left(
\begin{array}{cccccccccccccccccccccccccccccccccccc}
1&1&1&1&1&0&0&1&0&1&0&1&0&1&0&1&0&1&0&1&1&1&1&1&1&1&1&1&1&0\\
1&1&0&1&1&0&1&0&0&0&1&0&1&1&1&1&0&1&0&0&0&0&1&1&0&0&1&1&1&1\\
0&0&0&0&0&1&1&1&0&1&1&0&1&1&1&1&0&0&1&0&1&0&0&1&1&0&0&1&1&0\\
1&1&0&0&0&1&0&0&0&1&1&0&0&1&1&1&1&0&1&1&0&1&0&0&1&0&1&0&0&1\\
1&0&1&1&0&1&1&0&1&0&0&1&1&1&0&0&1&0&0&1&1&0&1&0&1&0&0&0&1&0\\
1&1&0&0&1&0&1&0&0&1&0&1&0&0&1&1&1&0&1&0&1&0&1&1&1&1&0&0&0&0\\
\end{array}
\right)
\\
A_{2,40}=&
\left(
\begin{array}{cccccccccccccccccccccccccccccccccccc}
1&0&1&0&0&0&0&0&0&0&0&0&0&0&0&0&1&1&0&1&1&1&1&0&0&1&1&1&1&1&1&1&1&1\\
0&0&1&1&0&0&1&1&1&0&1&1&1&1&1&0&0&0&1&0&0&1&0&1&0&1&0&0&0&0&1&1&0&1\\
0&0&0&0&0&0&1&0&1&0&0&1&1&0&0&1&0&0&1&1&1&0&0&1&1&0&1&1&1&1&1&1&0&1\\
1&0&0&0&1&1&0&0&1&1&0&0&1&1&1&0&1&1&0&1&1&0&0&0&1&0&0&1&1&1&0&1&0&0\\
0&1&0&0&1&1&1&1&0&0&1&1&0&1&1&0&1&1&1&1&1&1&1&0&0&1&0&0&0&0&1&0&1&0\\
0&1&1&0&1&0&0&0&1&1&1&1&0&0&0&0&1&0&1&1&0&0&1&1&0&1&1&0&0&1&1&1&0&0\\
\end{array}
\right)
\\
A'_{2,40}=&
\left(
\begin{array}{ccccccccccccccccccccccccccccccccccc}
0&1&1&0&0&0&1&1&1&0&1&1&0&1&0&0&1&0&1&0&1&0&1&1&1&1&1&1&1&1&1&1\\
1&0&1&1&1&1&0&1&0&1&1&1&0&0&1&1&1&1&0&0&0&1&1&1&1&0&1&0&0&0&1&0\\
1&0&1&0&0&1&1&1&0&1&1&0&1&0&1&0&0&1&1&1&0&1&0&1&0&1&0&0&0&0&1&1\\
1&1&1&0&1&0&0&0&1&1&1&0&1&1&1&0&1&0&1&1&0&1&0&1&1&0&1&1&0&0&1&0\\
1&0&0&1&0&1&0&1&0&0&1&1&1&1&1&1&1&0&0&1&0&0&0&0&1&0&1&1&1&0&0&1\\
0&1&1&0&0&1&1&1&0&1&1&1&1&1&0&1&0&1&0&1&1&1&1&0&1&1&0&1&1&0&0&1\\
1&0&1&0&1&0&1&1&1&1&0&0&0&0&0&1&0&1&0&1&1&0&1&1&1&0&0&1&0&0&0&0\\
0&1&1&0&0&1&0&1&0&0&1&0&0&1&1&0&1&1&0&1&1&0&0&0&0&1&0&0&1&1&1&0\\
\end{array}
\right)
\end{align*}
}
\end{center}
\caption{Matrices $A_{2,34}$, $A_{2,36}$, $A_{2,40}$ and $A'_{2,40}$}
\label{Fig:F2-2}
\end{figure}

%%%%%%%%%%%%%%%%%%%%%%%%%%%%%%%%%%%%%%%
\begin{figure}[thbp]
\begin{center}
%{\small
{\footnotesize
%{\scriptsize
%{\tiny
\begin{align*}
A_{3,34}=
\left(\begin{array}{ccccccccccccccccc}
1&1&2&0&1&2&0&2&2&2&0&1 \\
1&1&0&2&1&0&1&2&2&1&1&2 \\
1&0&0&1&2&1&1&2&1&0&1&0 \\
1&2&2&0&0&0&2&1&2&1&0&0 \\
2&2&1&1&2&2&1&1&2&2&1&0 \\
1&0&0&2&0&1&0&2&0&0&2&1 \\
2&0&0&0&1&1&2&0&0&0&1&2 \\
2&2&0&0&2&1&0&1&0&2&1&1 \\
1&0&0&2&0&1&1&0&0&2&0&2 \\
2&0&0&2&2&0&1&1&2&1&1&0 \\
1&2&0&2&1&0&1&1&0&2&1&1 \\
0&1&1&1&0&1&1&2&2&0&2&1 \\
0&2&2&2&0&1&2&2&0&2&1&1 \\
2&1&2&0&2&0&0&2&0&0&1&0 \\
0&1&1&2&0&0&0&0&0&1&1&1 \\
0&2&1&0&2&2&0&2&0&0&1&1 \\
0&1&2&2&1&2&0&0&2&1&2&2 \\
0&1&0&1&2&2&0&1&2&0&2&2 \\
2&0&2&0&2&0&0&1&0&1&2&2 \\
1&1&0&1&1&1&0&1&2&2&2&1 \\
0&0&1&1&0&2&1&2&1&2&1&0 \\
1&1&1&2&2&1&1&1&1&2&0&0 \\
\end{array}\right)
A_{3,37}=
\left(\begin{array}{ccccccccccccccccc}
1&2&0&2&2&2&0&0 \\
0&1&2&0&2&2&2&0 \\
0&0&1&2&0&2&2&2 \\
1&0&2&2&2&1&1&2 \\
1&1&2&0&2&0&0&1 \\
2&1&2&1&0&1&1&0 \\
0&2&1&2&1&0&1&1 \\
2&0&0&0&2&0&1&1 \\
2&2&1&2&0&1&1&1 \\
2&2&0&0&2&2&2&1 \\
2&2&0&2&0&1&0&2 \\
1&2&1&1&2&1&0&0 \\
0&1&2&1&1&2&1&0 \\
0&0&1&2&1&1&2&1 \\
2&0&1&0&2&0&2&2 \\
1&2&2&2&0&0&2&2 \\
1&1&1&0&2&1&2&2 \\
1&1&0&2&0&0&0&2 \\
1&1&0&1&2&1&2&0 \\
0&1&1&0&1&2&1&2 \\
1&0&0&2&0&2&1&1 \\
2&1&1&2&2&2&0&1 \\
2&2&2&0&2&1&0&0 \\
0&2&2&2&0&2&1&0 \\
0&0&2&2&2&0&2&1 \\
2&0&1&1&2&1&1&2 \\
1&2&2&2&1&0&0&1 \\
2&1&0&1&2&0&1&0 \\
0&2&1&0&1&2&0&1 \\
\end{array}\right)
\end{align*}
}
\end{center}
\caption{Matrices $A_{3,34}$ and $A_{3,37}$}
\label{Fig:F3}
\end{figure}
%%%%%%%%%%%%%%%%%%%%%%%%%%%%%%%%%%%%%%%%%%%%%%%%%%%%

%%%%%%%%%%%%%%%%%%%%%%%%%%%%%%%%%%%%%%%%%%%%%%%%%%%%
\end{landscape}

\end{document}